\documentclass[%
 reprint,
 amsmath,amssymb,
 aps,
]{revtex4-2}

\usepackage{graphicx}
\usepackage{dcolumn}
\usepackage{bm}
\usepackage{graphics}
\usepackage{float}

\newtheorem{theorem}{Theorem}[section]
\newtheorem{lemma}[theorem]{Lemma}

\newenvironment{proof}[1][Proof]{\begin{trivlist}
		\item[\hskip \labelsep {\bfseries #1}]}{\end{trivlist}}

\newenvironment{Even Case}[1][Even Case]{\begin{trivlist}
		\item[\hskip \labelsep {\bfseries #1}]}{\end{trivlist}}
\newenvironment{Odd Case}[1][Odd Case]{\begin{trivlist}
		\item[\hskip \labelsep {\bfseries #1}]}{\end{trivlist}}

\newcommand{\qed}{\nobreak \ifvmode \relax \else
	\ifdim\lastskip<1.5em \hskip-\lastskip
	\hskip1.5em plus0em minus0.5em \fi \nobreak
	\vrule height0.75em width0.5em depth0.25em\fi}

\usepackage{algorithm, algpseudocode}
\usepackage{fancyvrb,newverbs,xcolor}
\usepackage{colortbl}
\definecolor{cverbbg}{gray}{0.93}

\newenvironment{lcverbatim}
 {\SaveVerbatim{cverb}}
 {\endSaveVerbatim
  \flushleft\fboxrule=0pt\fboxsep=.5em
  \colorbox{cverbbg}{%
    \makebox[\dimexpr\linewidth-2\fboxsep][l]{\BUseVerbatim{cverb}}%
  }
  \endflushleft
}

\newverbcommand{\cverb}
  {\setbox\verbbox\hbox\bgroup}
  {\egroup\colorbox{cverbbg}{\box\verbbox}}

\begin{document}

\preprint{APS/123-QED}

\title{Normal Frequency method for finding the existence of maximum clique of the clique complex}
\author{Youngik Lee}
\affiliation{Brown University, Department of Physics, Box 1843, 182 Hope Street, Barus \& Holley 545, Providence, RI 02912, USA}
\begin{abstract}
Determining the existence of $k$-clique in the arbitrary graph is one of the NP problems.
We suggest a novel way to determine the existence of $k$-clique in the clique complex $G$ under specific conditions, by using the normal mode and eigenvector of the Laplace matrix of $G$.

\end{abstract}

\maketitle
\tableofcontents
\section{Introduction}
Determining the existence of $k$-clique in arbitrary graphs is one of the interesting questions. It classifies as an NP-complete problem in math.
Resonance is a physical behavior if the arbitrary system S has natural frequency $f$, when the system is subjected to an external vibration with frequency $f$, then the system S absorbs energy and starts to vibrate with gradually larger amplitudes over time.
In the system of string and nodes, which we can map as a graph, to calculate the natural frequency of the system, we can use the normal mode of the system as the natural frequency of the system.  \cite{1}
Therefore the basic idea of this research is based on 
the expectation of whether we can determine the substructure of the arbitrary graph by using the resonance effect.
For the arbitrary graph structure, if we map each node on the mass on the circle, and the edge as a spring, then we can calculate the normal mode of the system.
Then let us assume we have graphs $G_1$ and $G_2$. And then we can map those graph structures as an incidence matrix $M_1$ and $M_2$. We can calculate each normal mode frequency and normal mode, and by comparing the normal mode frequency we can determine the inclusion relationship between two arbitrary graphs.
\hfill

\section{Theory}

Let us $G$ is the arbitrary clique complex of two cliques $g_1$ and $g_2$.
Then we can follow the below algorithm.

\begin{itemize}
    \item [(1)] Change the graph $G$ to incidence matrix, $A$.
    \item [(2)] Calculate $A^2$. And if there is an element less than 1, make that (row, column) element on $A$ as zero.
    \item [(3)] Calculate $A^3$. And if there is an element less than 1, make that (row, column) element on $A$ as zero.
    \item [(4)] Repeat process (2)-(3) until we get the same input and output matrix.
    \item [(5)] Use normal mode analysis and get the maximum clique of the system.
\end{itemize}

Here step (1) is the update of the graph to matrix.
Step (2)-(4) is considered as a filtering process of cutting out the node which is unavailable to draw a closed path length of more than 3, from that node.
Step (5) is called normal mode analysis, which is the main process of this research.

\subsection{Normal mode analysis}
If we consider the graph node as a point mass with $m$, and the edge as a spring with sprint constant $k$ and put the node on the circle, then by using Eq. (\ref{eq1}) we can calculate the normal frequency of the system

\begin{equation}\label{eq1}
\det\left(\left(\frac{1}{2}\alpha_i k-mw^2\right)\mathbf{I}-\frac{1}{2}kA\right)=0
\end{equation}

Here $\alpha_i$ is the sum of $i$-th row of input matrix $A$.

For convenience let us set spring constant $k$ and mass $m$ as 1.

\begin{equation}\label{eq2}
\det\left((\alpha_i-w^2)\mathbf{I}-A\right)=0
\end{equation}

We can rewrite this equation with the Laplacian matrix $L_A$ for incidence matrix $A$,

\begin{equation}\label{eq3}
\left(\alpha_i\mathbf{I}-A\right)v_i=L_Av_i=\omega_i^2v_i=\lambda_i v_i
\end{equation}

Here $v_i$ is eigenvector, and $\lambda_i$ is eigenvalue.

\subsubsection{disconnected set of complete graph}
In the case of set of disconnected set of complete graph, we can get the available list of normal frequency $\omega_i$. If we descending order $\omega_i$, then the first value $\omega_1^2$ is the maximum clique size of $G$.

\subsubsection{clique complex with two cliques}
In case of clique complex with two cliques, we can get the available list of normal frequency $\omega_i$. If we descending order $\omega_i$, then the second value $\omega_2^2$ is the maximum clique size of $G$.

\subsection{Mathematical Proof}

\subsubsection{Disconnected set of complete graph}
According to \cite{2} we have equation,

\begin{eqnarray}\label{eq4}
\det\left(A+B\right)&=&\sum_r\sum_{\alpha,\beta}(-1)^{s(\alpha)+s(\beta)}\nonumber\\
&&\cdot\det(A[\alpha|\beta])
\det(B[\alpha|\beta])
\end{eqnarray}

Here $A$ and $B$ are $n$-square matrix, and $r$ is a integer from 0 to $n$.
And $\alpha$ and $\beta$ is increasing integer sequence of length $r$ chosen from 1 to n. $A[\alpha|\beta]$ is the $r$-square submatrix of $A$ lying in rows $\alpha$ and column $\beta$.
$B[\alpha|\beta]$ is the $(n-r)$-square submatrix of $B$ lying in comtemplementary rows $\alpha$ and comtemplementary column $\beta$. $s(\alpha)$, $s(\alpha)$ are the sum of sequence $\alpha$, $\beta$.

If we have arbitrary disconnected set of complete graph $G$, the incidence matrix $A$ is given as,

\begin{equation}\label{eq5}
A=\begin{pmatrix}
Q & \mathbf{0}_{Q\times W} & ... & \mathbf{0}_{Q\times V} \\ 
\\
\mathbf{0}_{W\times Q} & W & ... & \mathbf{0}_{W\times V} \\ 
\\
\mathbf{0}_{V\times Q} & \mathbf{0}_{V\times W}& ... & V
\end{pmatrix}
\end{equation}

Here $Q\times V$ means the matrix has the same row size as $Q$ and the same column size as $V$. For convenience, let us omit the size of the $\mathbf{0}$ and $\mathbf{1}$ matrix.
Then we can write down the Laplacina matrix of $G$ as,

\begin{equation}\label{eq6}
L_A=\begin{pmatrix}
Q & 0 & ... & 0 \\ 
\\
0 & W & ... & 0 \\ 
\\
0 & 0 & ... & V
\end{pmatrix}
=\begin{pmatrix}
Q & 0 \\ 
\\
0 & Z
\end{pmatrix}
\end{equation}

\begin{equation}\label{eq7}
B=\begin{pmatrix}
Q & 0 \\ 
\\
0 & 0
\end{pmatrix},\qquad
C=\begin{pmatrix}
0 & 0 \\ 
\\
0 & Z
\end{pmatrix}
\end{equation}

Then to calculate the determinant of $A$, let us first proof

\begin{equation}\label{eq8}
\det(A)=\det(B+C)=\det(Q)\det(Z)
\end{equation}

Let us define $n_A$=size($A$), which means $A$ is the $n_A$ times $n_A$ matrix.
From Eq. (\ref{eq4}), if we choose any all zero row or column then the result is going to be zero. Therefore for the available choice for $\alpha$ is $\{1,2,..,n_Q\}$. Vice versa for $B$ the available choice of $\beta$ is same as $\alpha$.

\begin{equation}\label{eq9}
\det(B+C)=(-1)^{s(\alpha)+s(\beta)}\det(B[\alpha|\beta])\det(C[\alpha|\beta])
\end{equation}

\begin{eqnarray}\label{eq10}
\det(A)&=&\det(B+C)=(-1)^{n_Q(n_Q+1)}\det(Q)\cdot\det(Z)\nonumber\\
&=&\det(Q)\cdot\det(Z)\qed
\end{eqnarray}

From Eq. (\ref{eq6}) we can rewrite matrix $Z$ as,

\begin{equation}\label{eq11}
Z=\begin{pmatrix}
W & 0 & ... & 0 \\ 
\\
0 & R & ... & 0 \\ 
\\
0 & 0 & ... & V
\end{pmatrix}
\end{equation}

And we can recursively using Eq. (\ref{eq10}),
\begin{eqnarray}\label{eq12}
\det(Z)&=&\det(W)\cdot...\cdot\det(V)
\end{eqnarray}

Therefore,

\begin{equation}\label{eq13}
\det(L_A)=\det(Q)\cdot\det(W)\cdot...\cdot\det(V)
\end{equation}

And if $Q$, $W$, ... $V$ is all complete graph's incidence matrix, by using Lemma II. 1 and Eq. (\ref{eq2})-(\ref{eq3}) one can be conclude that the maximum value for $\omega_i^2$ is the largest clique size in the graph. 

\begin{lemma}\label{lemma1}
The Laplacian for the complete graph $K_n$ on $n$ vertices has eigenvalue 0 with multiplicity 1 and eigenvalue $n$ with multiplicity $n-1$. 
\end{lemma}

\begin{proof}
By Lemma 3.1.1 on \cite{3}, eigenvalue 0 has multiplicity 1 with all-1s eigenvector $\mathbf{1}$. And if we take any non-zero vector, $v$ which is orthogonal to $\mathbf{1}$, 

\begin{equation}\label{eq14}
\mathbf{1}^{\rm T}\cdot v=\sum_i v_i=0
\end{equation}

\begin{equation}\label{eq15}
L_{K_n}=(n-1)\mathbf{I}_n-A,
\qquad A=\begin{cases}0&{\mbox{if}}\ i=j\\1&{\mbox{if}}\ i\neq j\ \end{cases}
\end{equation}

Note that this value is equal to,
\begin{equation}\label{eq16}
L_{K_n}v_i=(n-1)v_i-\sum_{i\neq j} v_j = nv_i-\sum_j v_j=nv_i
\end{equation}

Therefore any vector $v$ which is orthogonal to $\mathbf{1}$ is an eigenvector with an eigenvalue $n$. \qed
\end{proof}

\subsubsection{Clique-complex with two cliques}
We can express the Laplacian matrix as,

\begin{equation}\label{eq17}
L_A=\begin{pmatrix}
Q & -\mathbf{1}_{Q\times W} & \mathbf{0}_{Q\times V} \\ 
\\
-\mathbf{1}_{W\times Q} & W & -\mathbf{1}_{W\times V} \\ 
\\
\mathbf{0}_{W\times Q} & -\mathbf{1}_{V\times W} & V 
\end{pmatrix}
\end{equation}

Here $Q, W, V$ has a form of 

\begin{eqnarray}\label{eq18}
Q_{ij}&=&\begin{cases}n_{Q+W}&{\mbox{if}}\ i=j\\
-1&{\mbox{if}}\ i\neq j\ \end{cases}    \\
W_{ij}&=&\begin{cases}n_{Q+W+V}&{\mbox{if}}\ i=j\\
-1&{\mbox{if}}\ i\neq j\ \end{cases}    \\
V_{ij}&=&\begin{cases}n_{W+V}&{\mbox{if}}\ i=j\\
-1&{\mbox{if}}\ i\neq j\ \end{cases}
\end{eqnarray}

For convenience, let us omit the size of the $\mathbf{0}$ and $\mathbf{1}$ matrix. Then we can calculate the determinant of $A$ as,

\begin{equation}\label{eq21}
A=\begin{pmatrix}
Q & -1 & 0 \\ 
\\
-1 & W & -1 \\ 
\\
0 & -1 & V 
\end{pmatrix}
\end{equation}

\begin{equation}\label{eq22}
Av_i=\lambda_i v_i
\end{equation}

when we define eigenvector $v_1$ as

\begin{eqnarray}\label{eq23}
v_1^{\rm T}&&=\{\underbrace{0,0,...0}_{n_Q},|\underbrace{-1,0,..,0,1}_{n_W},|\underbrace{0,0,...,0}_{n_V}\}\nonumber\\
&&=\{q'|w'|v'\}    
\end{eqnarray}

and let us define the matrix $P,R,Z$ as,

\begin{equation}\label{eq24}
P=\begin{pmatrix}
Q & -1 & 0 \\ 
\\
-1 & W & 0 \\ 
\\
0 & 0 & 0 
\end{pmatrix}
\end{equation}

\begin{equation}\label{eq25}
R=\begin{pmatrix}
0 & 0 & 0 \\ 
\\
0 & W & -1 \\ 
\\
0 & -1 & V 
\end{pmatrix}
\end{equation}

\begin{equation}\label{eq26}
Z=\begin{pmatrix}
0 & 0 & 0 \\ 
\\
0 & W & 0 \\ 
\\
0 & 0 & 0 
\end{pmatrix}
\end{equation}

Then we can express $A=P+R-Z$

\begin{equation}\label{eq27}
Av_1=(P+R-Z)\cdot v_1=
\begin{pmatrix}
    -1_{Q\times W}\cdot w'\\
    W\cdot w'\\
    0
\end{pmatrix}
=\begin{pmatrix}
    0\\
    W\cdot w'\\
    0    
\end{pmatrix}
\end{equation}

\begin{eqnarray}\label{eq28}
W\cdot w'&=&    
\begin{pmatrix}
    n_{Q+W+V}-1&-1&...&-1\\
    -1&n_{Q+W+V}-1&...&-1\\
    -1& -1 &...&-1\\
    -1& -1 &...&n_W-1    
\end{pmatrix}\cdot 
\begin{pmatrix}
    -1\\
    0\\
    ...\\
    0\\
    1    
\end{pmatrix}\nonumber\\
&=&
\begin{pmatrix}
    -n_{Q+W+V}\\
    0\\
    ...\\
    0\\
    n_{Q+W+V}    
\end{pmatrix}=n_{Q+W+V}\cdot w
\end{eqnarray}

\begin{equation}\label{eq29}
Av_1=n_{Q+W+V}v_1=\lambda_1 v_1=\omega_1^2 v_1
\end{equation}

when we define eigenvector $v_2$ as

\begin{eqnarray}\label{eq30}
v_2^{\rm T}&&=\{\underbrace{0,0,...0}_{n_Q},|\underbrace{0,..,0}_{n_W},|\underbrace{-1,0,...,1}_{n_V}\}\nonumber\\
&&=\{q''|w''|v''\}    
\end{eqnarray}

\begin{equation}\label{eq31}
Av_2=(P+R-Z)\cdot v_2=
\begin{pmatrix}
    0\\
    -1_{W\times V}\cdot v''\\
    V\cdot v''
\end{pmatrix}
=\begin{pmatrix}
    0\\
    0\\
    V\cdot v''    
\end{pmatrix}
\end{equation}

Therefore,
\begin{eqnarray}\label{eq32}
V\cdot v''&=&    
\begin{pmatrix}
    n_{W+V}-1&-1&...&-1\\
    -1&n_{W+V}-1&...&-1\\
    -1& -1 &...&-1\\
    -1& -1 &...&n_{W+V}-1    
\end{pmatrix}\cdot 
\begin{pmatrix}
    -1\\
    0\\
    ...\\
    0\\
    1    
\end{pmatrix}\nonumber\\
&=&
\begin{pmatrix}
    -n_{W+V}\\
    0\\
    ...\\
    0\\
    n_{W+V}    
\end{pmatrix}=n_{W+V}\cdot v''
\end{eqnarray}

\begin{equation}\label{eq33}
Av_2=n_{W+V}v_2=\lambda_2 v_2=\omega_2^2 v_2
\end{equation}

Therefore the maximum frequency $\omega_1^2$ is the same as $n_A$, and the $\omega_2^2$ is the size of maximum clique of $G$.
\qed

\section{Conclusion and Outlook}

In this research, we proof the normal frequency analysis method to find the maximum clique of the particular type of graph. (disconnected set of complete graph, and clique-complex with two cliques).
For the next step, we can investigate,

\begin{itemize}
    \item[(1)] Find normal frequency analysis method for arbitrary clique complexes with multiple cliques.
    \item[(2)] Find a way to map arbitrary graphs to clique complex by updating the filtering process.
\end{itemize}

And if there is no particular mapping between arbitrary graphs to clique complex, further investigation to develop a new method is necessary.

\section{Appendix}

\subsection{Simulation Code}
The sample mathematical code for the Algorithm is introduced in Section 2.
\begin{lcverbatim}
A = {
   {0, 1, 1, 0, 0, 0, 0},
   {1, 0, 1, 1, 1, 1, 1},
   {1, 1, 0, 1, 1, 1, 1},
   {0, 1, 1, 0, 1, 1, 1},
   {0, 1, 1, 1, 0, 1, 1},
   {0, 1, 1, 1, 1, 0, 1},
   {0, 1, 1, 1, 1, 1, 0}
   };
B = Solve[
   Det[DiagonalMatrix[(k/2) 
   Total[A[[Range[n]]]]-m*w^2]-(k/2)*A] ==
     0, w];
DeleteDuplicates[B]
\end{lcverbatim}


\end{document}